\documentclass[12pt]{iopart}
\usepackage{iopams}
\usepackage{setstack}
\usepackage{graphicx}
\usepackage{amsthm}
\newtheorem{theorem}{Theorem}

\newtheorem{lemma}{Lemma}

\newtheorem{definition}{Definition}
\newtheorem{propose}{Proposition}
\theoremstyle{remark}

\begin{document}
\newcommand{\real}{\textrm{Re}\:}
\newcommand{\sto}{\stackrel{s}{\to}}
\newcommand{\supp}{\textrm{supp}\:}
\newcommand{\wto}{\stackrel{w}{\to}}
\newcommand{\ssto}{\stackrel{s}{\to}}
\newcounter{foo}
\providecommand{\norm}[1]{\lVert#1\rVert}
\providecommand{\abs}[1]{\lvert#1\rvert}

\title{Zero Energy Bound States in Many--Particle Systems}

\author{Dmitry K. Gridnev\footnote[1]{On leave from: Institute of Physics, St. Petersburg State
University, Ulyanovskaya 1, 198504 Russia}}

\address{FIAS, Ruth-Moufang-Stra{\ss}e 1, D--60438 Frankfurt am Main, Germany}
\ead{gridnev@fias.uni-frankfurt.de}
\begin{abstract}
It is proved that the eigenvalues in the N--particle system are absorbed at zero
energy threshold, if none of the subsystems has a bound state with $E \leq 0$
and none of the particle pairs has a zero energy resonance. The pair potentials are
allowed to take both signs.
\end{abstract}


\section{Introduction}\label{sec:1}

In \cite{1} it was proved that the 3--body system, which is at the 3--body
coupling constant threshold,
has a square integrable state at zero energy if none of the 2--body subsystems
is bound or has a zero energy resonance. The
condition on the absence of 2--body zero energy resonances is essential, that is
the 3--body ground state at zero energy can be at most a resonance and not an
$L^2$ state if at least one pair of particles has a zero energy resonance \cite{1}. One
of the restrictions on pair potentials in \cite{1} was their being
non--positive.
The aim of the present paper is to generalize the result of \cite{1} to the case
of
many particles and get rid of the restriction on the sign of  pair
potentials. The main result is expressed
in Theorems~\ref{th:1}, \ref{th:2}, which state that the eigenvalues in the N--particle system
are absorbed at zero energy threshold, if none of the subsystems has a bound
state
with $E \leq 0$ and none of the particle pairs has a zero energy resonance. Throughout
the paper we use the following operator notation.
$A\geq 0$ means that  $(f, Af) \geq 0$ for all $f \in D(A)$ and $A \ngeqq 0$
means that there exists $f_0 \in D(A)$
such that $(f_0, Af_0) < 0$.

We consider the $N$--particle Schr\"odinger operator 
\begin{equation}\label{hami}
    H(\lambda) = H_0 + \lambda \sum_{1 \leq i<j \leq N} V_{ij}(r_i - r_j),
\end{equation}
where $\lambda > 0$ is a coupling constant, $H_0$ is a kinetic energy operator
with the centre of mass removed,
$r_i \in \mathbb{R}^3$ are particle position vectors, the pair potentials are
real and $V_{ij} \in L^2 (\mathbb{R}^3) \cap L^1 (\mathbb{R}^3)$. The operator
$H(\lambda)$
is self--adjoint on $D(H_0) \subset  L^2 (\mathbb{R}^{3N-3})$, the set of
relative coordinates in $\mathbb{R}^{3N-3}$ we shall denote as $\xi$. Throughout the paper we shall assume that 
\begin{equation}
 \sigma_{ess} \bigl(H(\lambda)\bigr) = [0, \infty), 
\end{equation}
which, of course, restricts possible values of $\lambda$. Here we shall extensively use the term \textit{critical coupling}. In the literature 
one finds several related definitions: critically
bound \cite{richardmolec}, critical coupling \cite{richard}, coupling constant threshold in \cite{klaus1,klaus2},
virtual level at the threshold \cite{ahia}, etc. To avoid possible confusion 
we list three of the most popular definitions and indicate the relations between them. 

\begin{definition}\label{def1}
$ H(\lambda)$ is at critical coupling  if $H(\lambda) \geq 0$
and $H(\lambda) + \epsilon \sum_{i < j} V_{ij} \ngeqq 0$ for any $\epsilon >
0$.
\end{definition}
In the terminology of \cite{klaus1,klaus2} Def.~1 implies that $H(\lambda)$ is at the coupling constant threshold. 
The next definition due to scaling arguments is fully equivalent to Def.~1 
\begin{definition}\label{def2}
$ H(\lambda)$  is at critical coupling if $H(\lambda) \geq 0$
and $H(\lambda) - (1-\epsilon)H_0 \ngeqq 0$ for any $0 < \epsilon < 1$.
\end{definition}
So under the term \textit{critical coupling} we shall mean any of those. The next definition
can be found, for example, in \cite{ahia} 
\begin{definition}\label{def3}
$ H(\lambda)$ is said to have a virtual level at zero energy if
$H(\lambda) \geq 0$ and $H(\lambda) - \epsilon V_R \ngeqq 0$ for any $\epsilon >
0$, where $V_R := 1/(1 + |\xi|^2)$.
\end{definition}

In the case of $N=2$ it is easy to show that all three definitions are
equivalent to the definition of a two--particle zero energy resonance, c.f. \cite{sobolev,yafaev}. Note, that, in general, for $N \geq 3$ the Defs.~1--2 and Def.~3 are not
equivalent. The difference lies in the fact that the perturbation in Def.~3 does not move the lower bound
of the essential spectrum, since $V_R$ is a relatively
$H_0$--compact perturbation, contrary to the perturbations in Defs. 1--2, where
the lower bound of the essential spectrum can be moved, if some of the
subsystems are at critical coupling. 
\begin{propose}\label{prop:1}
A system of N particles is at critical coupling if it has a virtual level at zero energy. 
\end{propose}
\begin{proof}
Assume by contradiction that the system is not at critical coupling. Then there must exist $\epsilon_0 > 0$ such that $H - \epsilon_0 H_0
\geq 0$.
By the Courant identity \cite{courant,cikon}  there exists $\kappa > 0$ such that $H_0 - \kappa V_R
\geq 0$. Hence,
\begin{equation}
 H - \epsilon_0 \kappa V_R = H - \epsilon_0 H_0 + \epsilon_0 (H_0 - \kappa V_R )
\geq 0 ,
\end{equation}
which means that the system does not have a virtual level at zero energy. 
\end{proof}
As already mentioned the converse of Proposition~\ref{prop:1} is in general not true. 
Note, however, that that if a system has a zero energy bound state then it automatically has a virtual level at zero energy.

\section{Main Result}\label{sec:2}

For the formulation of Theorem~\ref{th:1} we need to impose the following requirement.
\begin{list}{R\arabic{foo}}
{\usecounter{foo}
    \setlength{\rightmargin}{\leftmargin}}
\item There exists a sequence of coupling constants
$\lambda_n \in \mathbb{R}_+$ such that  $\lim_{n\to \infty} \lambda_n =
\lambda_{cr} \in \mathbb{R}_+$, and
$H(\lambda_n) \psi_n = E_n \psi_n$, where $\psi_n \in D(H_0)$, $\|\psi_n\| =1$,
$E_n < 0$, $\lim_{n \to \infty} E_n = 0$.
\end{list}
Further in this section we shall prove the following
\begin{theorem}\label{th:1}
Suppose $H(\lambda)$ defined in (\ref{hami}) for $N \geq 3$  satisfies R1, $H(\lambda_n)$ and $H(\lambda_{cr})$ have no
subsystems, which have a bound state with $E \leq 0$, and no particle pairs at critical coupling. Then there
exists normalized $\psi_0 \in D(H_0)$ such that  $H(\lambda_{cr}) \psi_0 = 0$.
\end{theorem}
The next statement can be considered as a corollary to Theorem~\ref{th:1}. 
\begin{theorem}\label{th:2}
Suppose that $N \geq 3$ and $H(\lambda_{cr})$ is at critical coupling. Suppose also that $H(\lambda_{cr})$ has no
subsystems, which have a bound state with $E \leq 0$, and no particle pairs at critical coupling. 
Then there
exists normalized $\psi_0 \in D(H_0)$ such that  $H(\lambda_{cr}) \psi_0 = 0$.
\end{theorem}
\begin{proof}
Let us assume that none of the subsystems is at critical coupling. 
On one hand, from  the HVZ theorem
\cite{reed,cikon} it follows that there exists $\epsilon_0 >0 $
such that for $\lambda_n = \lambda_{cr} (1 + \epsilon_0 /n)$ and
$n = 1,2, \ldots$ we have $\inf \sigma_{ess} H(\lambda_n) = 0$.  We also choose $\epsilon_0$ small enough to guarantee that 
$H(\lambda_n)$ has no subsystems that are either bound or at critical coupling. On the other
hand, $H(\lambda_n) \ngeqq 0$. Therefore, there are
$\psi_n \in D(H_0)$ such that $H(\lambda_n) \psi_n = E_n \psi_n$, where $E_n <  0$,
$\|\psi_n \| = 1$ and $E_n \to 0$. Now the statement follows from Theorem~\ref{th:1}. It remains to get rid of the assumption that there is no 
subsystems at critical coupling. If there would be such then it is always possible to pass to the corresponding subsystem (call it $\mathcal{S}$), 
which has no subsystems at critical coupling. In such case by the above analysis $\mathcal{S}$ must have a bound state with $E=0$, which is in contradiction with the theorem conditions.  
\end{proof}

Following \cite{1} let us introduce the operator $B_{\tau_1 \tau_2} (z)$, where $1 \leq \tau_1 < \tau_2 \leq N$. 
We construct $B_{12}(z)$, for other particle pairs the construction is analogous.

We use Jacobi coordinates \cite{messiah} $\xi = (x, y_1, y_2, \ldots, y_{N-2})$,
where $x, y_i \in \mathbb{R}^3$. We set $x = \alpha^{-1} (r_2 - r_1)$ and $y_1 = (\sqrt{2 M_{12}}/\hbar)\bigl[ r_3 -
m_1/(m_1+m_2) r_1 - m_2/(m_1+m_2) r_2\bigr]$, 
where $\alpha := \hbar /\sqrt{2\mu_{12}} $, $M_{12} := (m_1 + m_2)m_3 / (m_1 + m_2 + m_3)$ and 
$\mu_{ik} := m_i m_k /(m_i + m_k)$ is the reduced mass. For $N=4$ this choice of coordinates is illustrated in Fig.~1 (Left). 
The coordinate $y_i \in \mathbb{R}^3$ is proportional to the
vector pointing from the centre of mass of the particles $[1, 2, \ldots , i+1]$
to the particle $i+2$, and the scale is set  to make the kinetic energy operator take the form
\begin{equation}\label{ay4}
    H_0 = - \Delta_x - \sum_i \Delta_{y_i} .
\end{equation}
Let $\mathcal{F}_{12}$ denote the partial Fourier transform in
$L^2(\mathbb{R}^{3N-3})$ acting as follows
\begin{equation}\label{ay5}
\fl \hat f(x,p_y)  =  \mathcal{F}_{12} f = \frac 1{(2 \pi )^{(3N-6)/2}} \int
d^{3N-6} y \;\; e^{-ip_y \cdot \; y} f(x,y) ,
\end{equation}
where $y = (y_1 , \ldots, y_{N-2}), \; p_y = (p_{y_1}, p_{y_2}, \ldots ,
p_{y_{N-2}}) \in \mathbb{R}^{3N-6}$.
Then $B_{12}(z)$ is defined through
\begin{equation}\label{ay6}
B_{12}(z) = 1 + z  + \mathcal{F}^{-1}_{12} t(p_y) \mathcal{F}_{12},
\end{equation}
where
\begin{equation}\label{ay669}
t(p_y) = \left(\sqrt{|p_y|} - 1\right)\chi_{\{p_y | \; |p_y| \leq 1\}}  , 
\end{equation}
$|p_y| = \bigl(\sum_i p_{y_i}^2 \bigr)^{1/2}$ and $\chi_\Omega$ denotes the characteristic function of the set $\Omega$.	
Let us transform the coordinates through
$\tilde y_i =\sum_k T_{ik} y_k$, where $T_{ik}$ is any orthogonal $(N-2)\times(N-2)$
matrix.
It is easy to check that
the construction of $B_{12} (z)$ is invariant with respect to these coordinate
transformations. That is 
\begin{equation}\label{invarianz}
 B_{12}(z) = 1 + z  + \tilde \mathcal{F}^{-1}_{12} t(\tilde p_y) \tilde \mathcal{F}_{12}  , 
\end{equation}
where $\tilde \mathcal{F}_{12}$ is defined through 
\begin{equation}\label{four2}
\fl \hat f(x,\tilde p_y)  =  \tilde \mathcal{F}_{12} f= \frac 1{(2 \pi )^{(3N-6)/2}} \int
d^{3N-6} \tilde y \;\; e^{-i \tilde p_y \cdot \; \tilde y} f(x,\tilde y) . 
\end{equation}
Similarly, one defines $B_{\tau_1 \tau_2} (z)$ for all particle pairs.
$B_{\tau_1 \tau_2} (z)$ and $B^{-1}_{\tau_1 \tau_2}  (z)$ are analytic on $\real z > 0$.

\begin{proof}[Proof of Theorem~\ref{th:1}]
By contradiction, let us assume that the zero energy bound state does not exist.
Then by Theorem~1 in \cite{1} $\psi_n$ totally spreads and $\psi_n \wto 0$ (for the definition of spreading see \cite{1}). 
Let $\tau = 1,2, \ldots, N(N-1)/2$ for $N \geq 4$ label  all particle pairs and $\tau_1 <
\tau_2$ label the particle numbers entering the pair $\tau$. We shall denote
$v_\tau := V_{\tau_1 \tau_2}$. It is helpful to split $v_\tau $ into positive
and negative parts $v_\tau = (v_\tau)_+ - (v_\tau)_-$, where $(v_\tau)_+ :=
\max[0, v_\tau]$ and
$(v_\tau)_- := \max[0, -v_\tau]$.
On one hand, the Schr\"odinger equation for $\psi_n$ reads
\begin{equation}\label{e0}
 \Bigl( H_0 + \lambda_n U_+ + k_n^2 \Bigr)\psi_n = \lambda_n \sum_\tau
\sqrt{(v_\tau)_-} \bigl( \sqrt{(v_\tau)_-} \psi_n \bigr) ,
\end{equation}
where we set
\begin{equation}\label{U+}
 U_+ := \sum_\tau (v_\tau)_+ .
\end{equation}
Acting on the last equation with an inverse operator gives
\begin{equation}\label{e1}
 \psi_n = \lambda_n \sum_\tau \Bigl( H_0 + \lambda_n U_+ + k_n^2 \Bigr)^{-1}
\sqrt{(v_\tau)_-} \Bigl( \sqrt{(v_\tau)_-} \psi_n \Bigr)  . 
\end{equation}
On the other hand, we can rearrange the terms in the Schr\"odinger equation as
follows
\begin{equation}\label{e2}
 \Bigl[ H_0 + \lambda_n(v_\tau)_+ + k_n^2 \Bigr] \psi_n = \lambda_n (v_\tau)_-
\psi_n - \lambda_n \sum_{\delta \neq \tau} v_\delta \psi_n ,
\end{equation}
where index $\delta$ runs through all particle pairs. This gives us
\begin{equation}\label{e3}
\fl  \psi_n = \lambda_n \Bigl[ H_0 + \lambda_n(v_\tau)_+ + k_n^2 \Bigr]^{-1}
(v_\tau)_- \psi_n - \lambda_n \sum_{\delta \neq \tau} \Bigl[ H_0 +
\lambda_n(v_\tau)_+ +k_n^2 \Bigr]^{-1} v_\delta \psi_n  . 
\end{equation}
Using (\ref{e3}) we get the following expression for the last term in brackets
in (\ref{e1})
\begin{eqnarray}
 \fl \sqrt{(v_\tau)_-} \psi(\lambda_n) = -\lambda_n \sum_{\delta \neq \tau} \Bigl\{
1 - \lambda_n \sqrt{(v_\tau)_-} \Bigl[ H_0 + \lambda_n(v_\tau)_+ + k_n^2
\Bigr]^{-1} \sqrt{(v_\tau)_-} \Bigr\}^{-1} \nonumber\\
\sqrt{(v_\tau)_- }\Bigl[ H_0 + \lambda_n(v_\tau)_+ +k_n^2 \Bigr]^{-1} v_\delta
\psi_n   . \label{e4b}
\end{eqnarray}
That the inverse of the operator in curly brackets makes sense would be shown in Lemma~\ref{lem:1} below.
Substituting (\ref{e4b}) into (\ref{e1}) yields the equation
\begin{eqnarray}
\fl \psi_n = -\lambda^2_n \sum_\tau \sum_{\delta \neq \tau} \Bigl( H_0 + \lambda_n
U_+ + k_n^2 \Bigr)^{-1} \sqrt{(v_\tau)_-} \Bigl\{  1 - \lambda_n \sqrt{(v_\tau)_-} \Bigl[ H_0 + \lambda_n(v_\tau)_+ +
k_n^2 \Bigr]^{-1} \nonumber\\
 \sqrt{(v_\tau)_-} \Bigr\}^{-1} \sqrt{(v_\tau)_- }\Bigl[ H_0 +
\lambda_n(v_\tau)_+ +k_n^2 \Bigr]^{-1} v_\delta \psi_n  \label{e5b}
\end{eqnarray}
All operators under the sum except $v_\delta$ are positivity preserving, see
\cite{lpestim,reed,cikon}. The inverse of the operator in curly brackets
being positivity preserving can be seen from its expansion in von Neumann
series, see and Lemma~12 in \cite{1} and Lemma~1 of this paper.
Thus we can transform (\ref{e5b}) into the following inequality
\begin{eqnarray}
\fl | \psi_n |\leq \lambda^2_n \sum_\tau \sum_{\delta \neq \tau} \Bigl( H_0 +
\lambda_n U_+ + k_n^2 \Bigr)^{-1} \sqrt{(v_\tau)_-} \Bigl\{  1 - \lambda_n \sqrt{(v_\tau)_-} \Bigl[ H_0 + \lambda_n(v_\tau)_+ +
k_n^2 \Bigr]^{-1}  \nonumber \\
\sqrt{(v_\tau)_-} \Bigr\}^{-1} \sqrt{(v_\tau)_- }\Bigl[ H_0 +
\lambda_n(v_\tau)_+ +k_n^2 \Bigr]^{-1}  |v_\delta| |\psi_n| .  \label{e6b}
\end{eqnarray}
Note that $\Bigl( H_0 + \lambda_n U_+ + k_n^2 \Bigr)^{-1} $ and $\Bigl[ H_0 +
\lambda_n(v_\tau)_+ + k_n^2 \Bigr]^{-1} $ are integral operators, see
\cite{lpestim}, and
positivity preserving operators, see, for example, \cite{reed} (Example 3 from
Sec.~IX.7 in vol. 2 and Theorem~XIII.44 in vol. 4).
By resolvent identities
\begin{eqnarray}
\fl  \Bigl( H_0 + k_n^2 \Bigr)^{-1} - \Bigl( H_0 + \lambda_n U_+ + k_n^2 \Bigr)^{-1}
= \lambda_n  \Bigl( H_0 + \lambda_n U_+ + k_n^2 \Bigr)^{-1} U_+   \Bigl( H_0 +
k_n^2 \Bigr)^{-1}  . \label{star1}\\
\fl \Bigl( H_0 + k_n^2 \Bigr)^{-1} - \Bigl( H_0 + \lambda_n(v_\tau)_+ + k_n^2
\Bigr)^{-1} = \lambda_n  \Bigl( H_0 + \lambda_n (v_\tau)_+  + k_n^2 \Bigr)^{-1}
(v_\tau)_+   \Bigl( H_0 + k_n^2 \Bigr)^{-1} .   \label{star2}
\end{eqnarray}
the differences on the lhs of (\ref{star1})--(\ref{star2}) are positivity
preserving operators. Therefore, we can rewrite (\ref{e6b}) as
\begin{eqnarray}
\fl | \psi_n |\leq \lambda^2_n \sum_\tau \sum_{\delta \neq \tau} \Bigl( H_0 + k_n^2
\Bigr)^{-1} \sqrt{(v_\tau)_-} \Bigl\{  1 - \lambda_n \sqrt{(v_\tau)_-} \Bigl[ H_0 + \lambda_n(v_\tau)_+ +
k_n^2 \Bigr]^{-1}  \nonumber\\
\sqrt{(v_\tau)_-} \Bigr\}^{-1} \sqrt{(v_\tau)_- }\Bigl[ H_0
+k_n^2 \Bigr]^{-1}  |v_\delta| |\psi_n| .  \label{e7b}
\end{eqnarray}

\begin{figure}[t]
\centering
\includegraphics[width=0.7\textwidth]{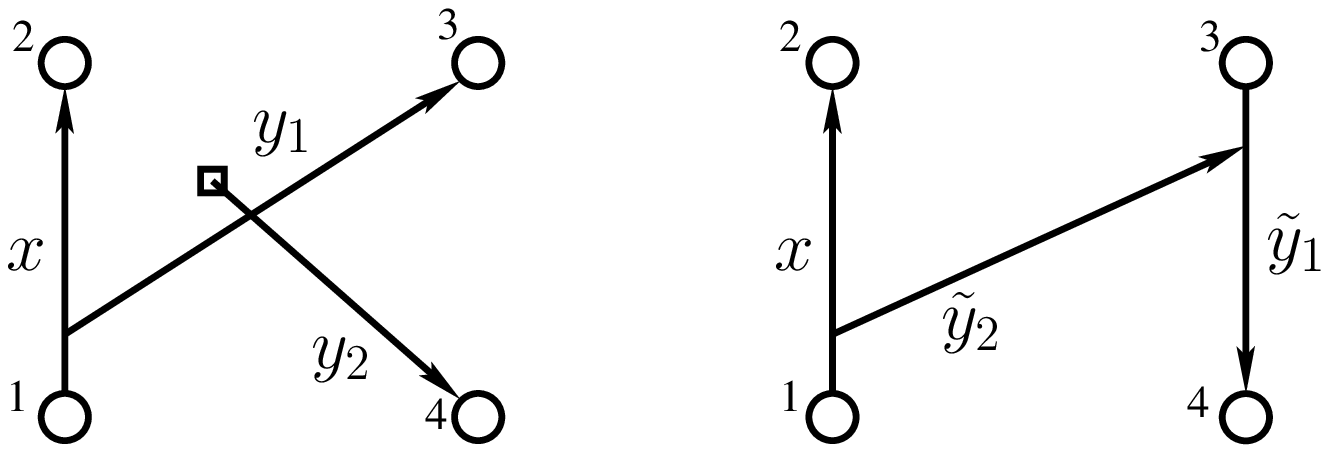}
\caption{Illustration to the choice of Jacobi coordinates for $N=4$. Left: $y_1$
points in the direction from the centre of mass of particles [1,2] to the particle
3 and $y_2$
points in the direction from the centre of mass of particles [1,2,3] (symbolized
by a square) to the particle 4. Right: $\tilde y_2$ points in the direction from
the centre
of mass of particles [1,2] to the centre of mass of particles [3,4]. The
coordinates' scales are set to make Eq.~(\ref{ay4}) hold.}
\label{fig:graph}
\end{figure}

We use the notation $B_\tau (z) \equiv B_{\tau_1 \tau_2} (z)$, where $B_{\tau_1 \tau_2} (z)$ was defined above. Inserting into (\ref{e7b}) the identity
$1 = B_\tau (k_n) B^{-1}_\tau (k_n)$ and using that $[B_\tau (k_n), H_0] = 0$
and $[B_\tau (k_n), (v_\tau)_\pm] = 0$ we obtain
\begin{equation}\label{e8}
 | \psi_n |\leq \lambda^2_n \sum_\tau \sum_{\delta \neq \tau} A_\tau (k_n)
R_\tau (k_n) D_{\tau ; \delta} (k_n) \sqrt{|v_\delta|} |\psi_n| ,
\end{equation}
where we defined the operators
\begin{eqnarray}
 A_\tau (k_n) := \Bigl( H_0 + k_n^2 \Bigr)^{-1} \sqrt{(v_\tau)_-} B_\tau(k_n)  , 
\label{e9a}\\
R_\tau (k_n) := \Bigl\{  1 - \lambda_n \sqrt{(v_\tau)_-} \Bigl[ H_0 +
\lambda_n(v_\tau)_+ + k_n^2 \Bigr]^{-1} \sqrt{(v_\tau)_-} \Bigr\}^{-1}  , 
\label{e9b}\\
D_{\tau ;\delta}  (k_n) := \sqrt{(v_\tau)_- }\Bigl[ H_0  +k_n^2 \Bigr]^{-1}
B^{-1}_\tau (k_n) \sqrt{|v_\delta|}   \quad \quad (\tau \neq \delta)   . \label{e9c}
\end{eqnarray}
Note that $H(\lambda_{cr})$ satisfies the conditions of Theorem~\ref{lem:5} in Appendix. It is easy to see that $(\psi_n, H(\lambda_{cr})\psi_n) \to 0$, where $\psi_n$ 
totally spreads. Hence, by Theorem~\ref{lem:5} we have $\| \sqrt{|v_\delta|} \psi_n\| = (\psi_n , |v_\delta| \psi_n)\to 0$. 
Applying  Lemmas~\ref{lem:1}, \ref{lem:2} to the rhs of (\ref{e8}) tells us that 
it goes to zero in norm, which is a
contradiction, since $\| \psi_n \| = 1$ by R1.
\end{proof}
\begin{lemma}\label{lem:1}
 The operators $A_\tau (k_n) , R_\tau (k_n)$ given by
(\ref{e9a})--(\ref{e9b}) are uniformly norm--bounded.
\end{lemma}
\begin{proof}
 Without loosing generality we can consider the pair $\tau = (1,2)$, where
$\tau_1 = 1$ and $\tau_2 = 2$. The proof that $\|A_{12} (k_n) \|$ is uniformly
bounded follows the same pattern as the proof of Lemma~6 in \cite{1} and we
omit it here. The proof for $R_\tau (k_n)$ uses the Birman--Schwinger principle in the form suggested in \cite{ges-simon}.
Note that for self--adjoint operators $\mathcal{A}, \mathcal{B} \geq 0$, where $\mathcal{A}^{-1}$ and $\mathcal{A}^{-1/2} \mathcal{B}^{1/2}$ are bounded,
one has 
\begin{equation}\label{stead}
\|\mathcal{A}^{-1/2} \mathcal{B}  \mathcal{A}^{-1/2}\| = \| \mathcal{B}^{1/2} \mathcal{A}^{-1} \mathcal{B}^{1/2} \| , 
\end{equation}
which follows from $\| \mathcal{C}^\dagger \mathcal{C} \| = \|\mathcal{C} \mathcal{C}^\dagger  \|$ for any bounded
 $\mathcal{C}$, see f.~e. \cite{traceideals}.
Due to conditions of Theorem~\ref{th:1} there exists $1 > \omega' > 0$ (independent of $n$)
such that $ (1-\omega') H_0 + \lambda_n (v_\tau) \geq 0 $, or, equivalently 
\begin{equation}\label{taridi}
H_0 + \lambda_n (v_\tau) \geq \omega' H_0 . 
\end{equation}
By standard estimates there must exist $\gamma_0 > 0$ such that $H_0 - \gamma_0 \lambda_n (v_\tau)_- \geq 0$ for all $n$. 
Together with (\ref{taridi}) this means that there exists $\omega > 0$ independent of $n$ such that 
\begin{equation}
H_0 + \lambda_n (v_\tau) - \lambda_n
\omega (v_\tau)_- \geq 0  . 
\end{equation}
To use the identity (\ref{stead}) let us set
\begin{eqnarray}
 \mathcal{A}:= H_0 + \lambda_n (v_\tau)_+ + k_n^2  ,  \\
\mathcal{B}:= \lambda_n (v_\tau)_-  . 
\end{eqnarray}
Because $\mathcal{A} - (1+\omega )\mathcal{B} \geq 0$ for any $\phi \in D(H_0)$ we have
\begin{eqnarray}
(\phi, [\mathcal{A} - (1+\omega )\mathcal{B}]\phi) = (\tilde \phi, \bigl\{1 - (1+\omega
)\mathcal{A}^{-1/2}\mathcal{B}\mathcal{A}^{-1/2}\bigr\} \tilde \phi)\geq 0  ,
\end{eqnarray}
where $\tilde \phi := \mathcal{A}^{1/2}\phi$. For $\phi \in D(H_0)$ the functions $\tilde \phi$ span a dense set in $L^2 (\mathbb{R}^{3N-3})$
since $D(H_0)$ is dense and $\mathcal{A}^{-1/2}$ is bounded. Hence, $(1+\omega )\mathcal{A}^{-1/2}\mathcal{B}\mathcal{A}^{-1/2} \leq 1$ and
$\|\mathcal{A}^{-1/2}\mathcal{B}\mathcal{A}^{-1/2}\| \leq 1/(1+\omega )$. By identity (\ref{stead}) we obtain
\begin{equation}
\fl  \|\lambda_n \sqrt{(v_\tau)_-} \Bigl[ H_0 + \lambda_n(v_\tau)_+ + k_n^2
\Bigr]^{-1} \sqrt{(v_\tau)_-}\| = \| \mathcal{B}^{1/2} \mathcal{A}^{-1} \mathcal{B}^{1/2} \| \leq
1/(1+\omega),
\end{equation}
which means that $R_\tau (k_n) $ in (\ref{e9c}) is correctly defined and uniformly norm--bounded.
\end{proof}

\begin{lemma}\label{lem:2}
 The operators $D_{\tau ;\delta} (k_n) $ given by (\ref{e9c}) are uniformly
norm--bounded.
\end{lemma}
\begin{proof}
 Again it suffices to consider the pair $\tau = (1,2)$, where $\tau_1 = 1$ and
$\tau_2 = 2$. We split $D_{12 ;\delta}  (k_n)$ as follows
\begin{eqnarray}
D_{12 ;\delta}  (k_n) =  D^{(1)}_{12;\delta}  (k_n) + D^{(2)}_{12 ;\delta}
(k_n) , \\
 D^{(1)}_{12;\delta}  (k_n) := \sqrt{(v_{12})_- }\Bigl[ H_0  +k_n^2 \Bigr]^{-1}
\left\{ B^{-1}_{12} (k_n) - \frac 1{k_n +1} \right\} \sqrt{|v_\delta|}   , \\
D^{(2)}_{12 ;\delta}  (k_n) := (k_n +1)^{-1} \sqrt{(v_{12})_- }\Bigl[ H_0  +k_n^2 \Bigr]^{-1}
\sqrt{|v_\delta|}   . \label{secop}
\end{eqnarray}
For the operator in (\ref{secop}) we get (see Eqs.~(43)--(44) in \cite{1})
\begin{equation}
\fl  \Bigl\| D^{(2)}_{12 ;\delta}  (k_n)\Bigr\| \leq  \Bigl\| \sqrt{(v_{12})_- }\Bigl[ H_0  +k_n^2 \Bigr]^{-1} \sqrt{(v_{12})_- }\Bigr\|^{1/2} \; 
 \Bigl\| \sqrt{|v_\delta|}  \Bigl[ H_0  +k_n^2 \Bigr]^{-1} \sqrt{|v_\delta|}  \Bigr\|^{1/2} , 
\end{equation}
where both norms in the product are uniformly bounded (this can be easily shown after making an appropriate Fourier transform). It remains
to prove that $ D^{(1)}_{12;\delta}   (k_n)$ is uniformly norm--bounded.
Let us first consider two cases: (a) $\delta_1 = 2$, $\delta_2 = 3$ and
(b) $\delta_1 = 3$, $\delta_2 = 4$. The proof for the case (a) almost repeats
the one in Lemma~9 in \cite{1}. Indeed, we need to show that $\|\mathcal{K}_n\|$ is uniformly bounded, where
\begin{equation}
\mathcal{K}_n  =   \mathcal{F}_{12}  D^{(1)}_{12;23}  (k_n)\mathcal{F}^{-1}_{12} .
\end{equation}
For convenience we denote $p_{y_r} := (p_{y_2} , p_{y_3}, \ldots, p_{y_{N-2}}) \in \mathbb{R}^{3N-9}$.
The integral operator $\mathcal{K}_n$ acts on $\phi(x,  p_{y_1},  p_{y_r} ) \in L^2(\mathbb{R}^{3N-3})$ as follows
\begin{equation}
\fl \mathcal{K}_n  \phi (x,  p_{y_1}, p_{y_r} ) = \int d^3 x' \: d^3
p'_{y_1} \: K_n (x, x' ,   p_{y_1},  p'_{y_1} ;  p_{y_r} )
\phi(x', p'_{y_1}, p_{y_r} )   ,
\end{equation}
where the integral kernel has the form \cite{1}
\begin{eqnarray}
\fl K_n (x,x', p_{y_1} , p'_{y_1}; p_{y_r}) = \frac 1{2^{7/2}\pi^{5/2} \gamma^3} \left[\frac 1{k_n + 1 +
t(p_y)} - \frac 1{k_n + 1} \right] \Bigl| \bigl( V_{12}\bigr)_-  (\alpha x)\Bigr|^{1/2} \nonumber \\
\times \frac{e^{-\sqrt{p_y^2 + k_n^2} |x-x'|}}{|x-x'|}
\exp{\left\{i\frac{\beta}{\gamma} x' \cdot (p_{y_1} - p'_{y_1})\right\}}\:
\widehat{\bigl| V_{23}\bigr|^{1/2}} ((p_{y_1} - p'_{y_1})/\gamma) ,  \label{four27}
\end{eqnarray}
$\beta := -m_2 \hbar / ((m_1 + m_2)\sqrt{2\mu_{12}})$
and $\gamma := \hbar/\sqrt{2M_{12}}$. Using the estimate 
\begin{equation}
  \| \mathcal{K}_n \|^2 \leq \sup_{ p_{y_r}} \int d^3 x \:  d^3 x' \: d^3
p_{y_1} \: d^3  p'_{y_1} \: \bigl|K_n (x, x' ,   p_{y_1},
p'_{y_1} ;  p_{y_r} )\bigr|^2
\end{equation}
we get 
\begin{equation}\label{wegot}
  \| \mathcal{K}_n \|^2 \leq C_0 \sup_{ |p_{y_r}| \leq \sqrt{1-p_{y_1}^2}} \int_{|p_{y_1}| \leq 1} d^3 p_{y_1} \left[\frac 1{k_n + \sqrt{|p_y|}}  - \frac 1{k_n +1} \right]^2 \frac 1{\sqrt{p_y^2 +k_n^2}} , 
\end{equation}
where the constant 
\begin{equation}
  \fl C_0 := \frac 1{2^7 \pi^5 \gamma^6} \left(\int d^3 x \; \bigl( V_{12}\bigr)_-  (\alpha x)\right) \left(\int d^3 s \;  \Bigl|\widehat{\bigl| V_{23}\bigr|^{1/2} } (s/ \gamma) \Bigr|^2 \right) 
\left(\int d^3 t \; |t|^{-2} e^{-2|t|}\right) 
\end{equation}
is finite. Continuing (\ref{wegot}) 
\begin{equation}
  \| \mathcal{K}_n \|^2 \leq C_0  \int_{|p_{y_1}| \leq 1} d^3 p_{y_1} \frac 1{|p_{y_1}|^2} = 4\pi C_0  . 
\end{equation}
In the case (b) we make the orthogonal transformation of Jacobi coordinates,
where $\tilde y_1  = \alpha'^{-1} (r_4 - r_3)$ and $\alpha' := \hbar
/\sqrt{2\mu_{34}} $.
Other transformed coordinates we denote as $\tilde y_r := (\tilde y_2, \ldots ,
\tilde y_{N-2}) \in \mathbb{R}^{3N-9}$.  This choice of coordinates for $N=4$ is
illustrated
in Fig.~1 (Right). We need to prove that $\|\mathcal{L}_n \|$ is uniformly bounded, where
\begin{equation}
\mathcal{L}_n =   \tilde \mathcal{ F}_{12} \sqrt{(v_{12})_- }\Bigl[ H_0  +k_n^2
\Bigr]^{-1}  \Bigl\{ B^{-1}_{12} (k_n) - 1 \Bigr\} \sqrt{|v_{34}|}
\tilde \mathcal{ F}^{-1}_{12}
\end{equation}
and $\tilde \mathcal{F}_{12}$ is defined as in (\ref{four2}). The operator $\mathcal{L}_n $ acts on
$\phi(x, \tilde p_{y_1}, \tilde p_{y_r} ) \in L^2(\mathbb{R}^{3N-3})$ as
\begin{equation}
\fl \mathcal{L}_n  \phi (x, \tilde p_{y_1}, \tilde p_{y_r} ) = \int d^3 x' \: d^3
p'_{y_1} \: L_n (x, x' ,  \tilde p_{y_1}, \tilde p'_{y_1} ; \tilde p_{y_r} )
\phi(x', \tilde p'_{y_1}, \tilde p_{y_r} )   ,
\end{equation}
where the integral kernel is
\begin{eqnarray}
\fl L_n (x, x' ,  \tilde p_{y_1}, \tilde p'_{y_1}	 ; \tilde p_{y_r} ) = \frac
1{2^{7/2}\pi^{5/2} (\alpha')^3} \left\{\frac 1{k_n + 1 + t(\tilde p_y)} - \frac
1{k_n + 1} \right\}
\Bigl| \bigl( V_{12}\bigr)_-  (\alpha x)\Bigr|^{1/2}  \nonumber \\
\times  \frac{e^{-\sqrt{\tilde p_y^2 + k_n^2} |x-x'|}}{|x-x'|}
\widehat{\bigl| V_{34}\bigr|^{1/2} } ((\tilde p_{y_1} - \tilde p'_{y_1})/\alpha') .
\end{eqnarray}
Now the proof that $ \|\mathcal{L}_n \|$ is uniformly bounded is identical to the one in the case (a) and so we omit it. The general case of $D^{(2)}_{12 ;\delta} $ follows from (a) and (b)
by making an orthogonal coordinate transformation, which corresponds to
the appropriate permutation of the particle numbers.
\end{proof}

\ack 

The author would like to thank Prof. Walter Greiner for the warm hospitality at FIAS.

\appendix
\section{The No-Clustering Theorem}\label{sec:3}

Below we prove the statement, which we call the no--clustering theorem. In the following $\chi_L : \mathbb{R}^3 \to \mathbb{R}$ denotes the function such that 
$\chi_L (r) = 1$ if $0 \leq |r| \leq L$ and zero otherwise.

We shall make use of the following Lemma concerning minimizing sequences
\cite{zhizhenkova}.
\begin{lemma}[(Zhislin)]\label{lem:6}
Suppose that $H \geq 0$ is given by (\ref{hami}), where $\lambda = 1$ and $V_{ij} \in L^2 +
L^\infty_\infty$. Suppose additionally that there is a normalized minimizing sequence $f_n \in D(H_0)$ 
such that $(f_n , H f_n) \to 0$. If $f_n$ does not totally spread then 
there exists a normalized $\phi_0 \in D(H_0)$ such that $H\phi_0 = 0$.
\end{lemma}
\begin{proof}
Since $f_n$ does not totally spread there must exist a subsequence such that $\| \chi_{\{x| |x| \leq R\}} f_{n_k}\| >a$
for some $R >0$ and $a>0$. We can assume that $f_{n_k} \wto \phi \in L^2 (\mathbb{R}^{3N-3})$ otherwise we could pass to the 
weakly converging subsequence, which exists by the Banach--Alaoglu theorem.
Thus for any $g \in D(H_0)$ we have
\begin{equation}\label{e9}
 (Hg, \phi ) = \lim_{k \to \infty}(Hg, f_{n_k}) = \lim_{k \to \infty}(H^{1/2}g,
H^{1/2}f_{n_k}) = 0
\end{equation}
because $\| H^{1/2}f_{n_k}\| = (f_{n_k} , H f_{n_k})  \to 0$ by condition of the lemma. From (\ref{e9}) it follows that $\phi \in D(H_0)$ and $H \phi = 0$. That
$\|\phi\| \neq 0$ follows from Lemma~3 in \cite{1}. Setting $\phi_0 = \phi/\|\phi\|$ we prove the lemma.
\end{proof}

\begin{theorem}\label{lem:5}
Suppose that $H$ is given by (\ref{hami}), where $\lambda = 1$ and $V_{ij} \in L^2 +
L^\infty_\infty$, and none of the subsystems has an eigenstate with an energy
less or equal to zero.
Let $\psi_n \in D(H_0)$ be a totally spreading sequence such that $(\psi_n, H
\psi_n) \to 0$. Then $(\psi_n, F(r_i - r_j ) \psi_n ) \to 0$ for all particle
pairs $(i,j)$ and
any given $F \in L^2 (\mathbb{R}^3) + L^\infty_\infty (\mathbb{R}^3)$.
\end{theorem}
\begin{proof}
Note that $\| H_0 \psi_n \|$ is uniformly bounded, c.~f. Lemma~1 in \cite{1} and $\psi_n \wto 0$ because $\psi_n$ totally spreads. By Lemma~\ref{lem:7} it is enough 
to prove the statement for $F(r) = \chi_L (r)$ and all $L >0$. 
For $N=2$ the statement becomes trivial. For $N \geq 3$ we prove the theorem by induction assuming that 
it holds for $N-1$ particles. Without loosing generality it is enough to show that $(\psi_n, \chi_L (r_1 - r_2 ) \psi_n ) \to 0$ for all $L>0$. 

We can assume that $\psi_n \in C_0^\infty (\mathbb{R}^{3N-3})$ otherwise we can pass to an appropriate sequence using that $C_0^\infty (\mathbb{R}^{3N-3})$ is dense in 
$D(H_0)$, see \cite{loss}. A proof by contradiction. Lets us assume that 
\begin{equation}\label{quist}
 \limsup_{n \to \infty} (\psi_n,\chi_L (r_1 - r_2 )\psi_n ) = a'
\end{equation}
for some $L >0$ and $a' >0$.

Let $J_s \in C^{2}
(\mathbb{R}^{3N-3})$ denote the Ruelle--Simon partition of unity, see Definition
3.4 and Proposition 3.5 in \cite{cikon}. For $s=1,2,\ldots, N$ one has
$J_s \geq 0$, $\sum_s J^{2}_s =1$ and $J_s (\lambda x) =J_s (x)$ for $\lambda
\geq 1$ and $|x|= 1$. Besides there exists $C > 0$ such that for $i \neq s$
\begin{equation}\label{ims5}
    \supp J_s \cap \{ x | |x| > 1 \} \subset \{x|\; |r_i - r_s | \geq C |x|\} .
\end{equation}
By the IMS formula (Theorem 3.2 in \cite{cikon})
\begin{equation}\label{ims}
    H = \sum_{s=1}^N J_s H_s  J_s + K , 
\end{equation}
where
\begin{eqnarray}
K := \sum_{s}\sum_{l \neq s}   V_{ls}  |J_s |^2 + \sum_s|\nabla J_s |^2  ,  \label{K}\\
H_s  := H - \sum_{l\neq s}V_{ls} . 
\end{eqnarray}
The operator $K$ is relatively $H_0$--compact, see Lemma~7.11 in \cite{teschl}. $H_s$ is the same operator as $H$ except that the pair--interactions 
that involve particle $s$ are switched off. 
By (\ref{quist}) we get 
\begin{equation}
\limsup_{n \to \infty} \sum_s (\psi^{(s)}_n, \chi_L (r_1 - r_2 ) \psi^{(s)}_n ) = a' , 
\end{equation}
where we define $\psi^{(s)}_n := J_s\psi_n \in C_0^\infty (\mathbb{R}^{3N-3})$. The operators $J_1 \chi_L (r_1 - r_2 )$ and $J_2 \chi_L (r_1 - r_2 )$ are relatively $H_0$ 
compact, hence, $\|\chi_L (r_1 - r_2 ) \psi^{(1)}_n \| \to 0$ and $\| \chi_L (r_1 - r_2 ) \psi^{(2)}_n \| \to 0$ by Lemma~2 in \cite{1}. Thus there must exist $s_0 \geq 3$ such that 
\begin{equation}\label{gre1}
\limsup_{n \to \infty}  (\psi^{(s_0)}_n, \chi_L (r_1 - r_2 ) \psi^{(s_0)}_n ) = 2a , 
\end{equation}
where $0 <a \leq 1/2$ is a constant. Let $\zeta \in \mathbb{R}^{3N-6}$ denote the internal Jacobi coordinates for the particles $\{1, 2, \ldots s_0 - 1, s_0+1, \ldots, N\}$ and 
$y \in \mathbb{R}^{3}$ the coordinate, which points from the particle $s_0$ to the center of mass of other particles. We choose the scales so that 
$H_0 = -\Delta_\zeta - \Delta_y$. 
It is convenient to introduce 
\begin{eqnarray}\label{fff}
    H^{(s_0)} := -\Delta_\zeta + V(\zeta) \nonumber  ,  \\
V(\zeta) := \sum_{i < k}V_{ik} - \sum_{l\neq s_0}V_{ls_0}  . 
\end{eqnarray}
Clearly, 
\begin{equation}\label{cmpar}
H_{s_0} \geq H^{(s_0)} \geq 0 . 
\end{equation}
The operator $ H^{(s_0)} $ is the Hamiltonian of the particles $\{1, 2, \ldots s_0 - 1, s_0+1, \ldots, N\}$ and 
can be considered on the domain $D(-\Delta_\zeta) \subset L^2(\mathbb{R}^{3N-6})$ as well. We have $\psi_n \wto 0$ because $\psi_n$ totally spreads. 
Because $K$ in (\ref{K}) is relatively $H_0$ compact we have $K \psi_n \to 0$, see Lemma~2 in \cite{1}. Using $(\psi_n ,  H \psi_n) \to 0$ and $H_s \geq 0$ 
we infer from (\ref{ims}) that $(\psi_n^{(s)}, H_s \psi_n^{(s)}) \to 0$ for all $s$. Hence, by (\ref{cmpar})
\begin{equation}\label{gre2}
(\psi_n^{(s_0)}, H^{(s_0)} \psi_n^{(s_0)}) \to 0  . 
\end{equation}
Looking at (\ref{gre1}) and (\ref{gre2}) we conclude that there exists a subsequence $\psi^{(s_0)}_{n_k}$ such that 
\begin{eqnarray}
(\psi^{(s_0)}_{n_k}, \chi_L (r_1 - r_2 ) \psi^{(s_0)}_{n_k}) \geq a , \label{gre4}\\
 (\psi^{(s_0)}_{n_k}, H^{(s_0)} \psi^{(s_0)}_{n_k}) \to 0 . \label{gre5}
\end{eqnarray}
From (\ref{gre4}) it follows that $\sqrt a \leq \| \psi^{(s_0)}_{n_k} \| \leq 1$. Thus defining $g_k := \psi^{(s_0)}_{n_k} / \| \psi^{(s_0)}_{n_k} \| $ we obtain 
\begin{eqnarray}
(g_k , \chi_L (r_1 - r_2 )g_k) \geq a  \label{gre6}\\
 \varepsilon_k := (g_k , H^{(s_0)} g_k ) \to 0 \label{gre7} , 
\end{eqnarray}
where $g_k \in C_0^\infty (\mathbb{R}^{3N-3})$ and $\|g_k \| = 1$. 
For $f(\zeta, y), h(\zeta, y) \in L^2(\mathbb{R}^{3N-3})$ let us introduce the notation 
\begin{equation}
 (f, h)_\zeta := \int d^{3N-6} \zeta \; f^* (\zeta, y) h(\zeta, y) , 
\end{equation}
where $(f, h)_\zeta$ depends on $y \in \mathbb{R}^3$. 

Now we define the following subsets of $\mathbb{R}^3$
\begin{eqnarray}
 \fl \mathcal{M}_k := \left\{ y \Bigl| (g_k , g_k)_\zeta  > 0 \right\} 
\cap \left\{ y \Bigl| (g_k, H^{(s_0)}
g_k )_\zeta < \sqrt{\varepsilon_k} (g_k, g_k)_\zeta \right\}  ,  \label{schra0}\\
\fl \mathcal{N}_k :=  \left\{ y \Bigl| (g_k, \chi_L (r_1 -r_2) g_k)_\zeta \geq (a/2) (g_k, g_k)_\zeta \right\}  . \label{schra1}
\end{eqnarray}
By standard results
$\mathcal{M}_k, \mathcal{N}_k$ are Borel sets. Below we prove that there exists
$k_0 $
such that $\mathcal{N}_k \cap \mathcal{M}_k \neq \emptyset$ for $k \geq k_0$.
For any Borel set $X \subset \mathbb{R}^3$ we define
\begin{equation}
 \mu_k (X) := \int _X d^3y \; (g_k, g_k)_\zeta  . 
\end{equation}
Because $g_k$ is normalized we have $\mu_k (\mathbb{R}^3) = 1$. On one
hand, using (\ref{gre7}) and (\ref{schra0})
\begin{equation}
 \fl \mu_k  (\mathbb{R}^3 / \mathcal{M}_k) = \int_{\mathbb{R}^3 / \mathcal{M}_k} d^3y \; 
(g_k , g_k )_\zeta 
\leq \frac 1{\sqrt{\varepsilon_k}} \int_{\mathbb{R}^3 / \mathcal{M}_k} d^3y \; (g_k, H^{(s_0)}g_k )_\zeta \leq \sqrt{\varepsilon_k}  . 
\end{equation}
Hence,
\begin{equation}
 \mu_k  (\mathcal{M}_k) \geq 1 - \sqrt{\varepsilon_k} . \label{schra3}
\end{equation}
On the other hand, using that according to (\ref{gre6}) $\int d^3y \; (g_k, \chi_L (r_1 - r_2) g_k)_\zeta \geq a$ we get
\begin{eqnarray}
 \fl \mu_k (\mathcal{N}_k) \geq \int_{\mathcal{N}_k} d^3y (g_k,
\chi_L(r_1-r_2)g_k)_\zeta \geq a - \int_{\mathbb{R}^3/
\mathcal{N}_k} (g_k,\chi_L(r_1-r_2)g_k)_\zeta  \nonumber\\
\geq a - \frac a2 \mu_k (\mathbb{R}^3/ \mathcal{N}_k) \geq  \frac a2  , 
\label{schra4}
\end{eqnarray}
where we applied (\ref{schra1}) and $\mu_k (\mathbb{R}^m/ \mathcal{N}_k) \leq
1$. Now it is clear that that there exists $k_0 $
such that $\mathcal{N}_k \cap \mathcal{M}_k \neq \emptyset$ for $k \geq k_0$.
Otherwise, according to (\ref{schra3}) and (\ref{schra4}) we
would have
\begin{equation}
 1 = \mu_k (\mathbb{R}^m) \geq \mu_k (\mathcal{M}_k) + \mu_k (\mathcal{N}_k)
\geq 1+ \frac a2 - \sqrt{\varepsilon_k},
\end{equation}
which is a contradiction since $\varepsilon_k \to 0$. Now we construct the minimizing
sequence for $H^{(s_0)}$ (considered now on $D(-\Delta_\zeta)$) taking any $y_k \in \mathcal{N}_k
\cap \mathcal{M}_k$ for $k \geq k_0$ and setting 
\begin{equation}
\phi_k (\zeta) := g_k (y_k , \zeta) \left(\int d^{3N-6}\zeta \; |g_k (y_k , \zeta)|\right)^{-1/2} . 
\end{equation}
Due to (\ref{schra0})--(\ref{schra1}) the sequence $\phi_k (\zeta) \in C_0^\infty (\mathbb{R}^{3N-6})$
has the following properties: $\|\phi_k \| =1$, $(\phi_k , H^{(s_0)} \phi_k) \to 0$ and 
\begin{equation}\label{mishka}
 (\phi_k , \chi_L (r_1 -r_2) \phi_k) \geq a/2. 
\end{equation}
By Lemma~\ref{lem:6} $\phi_k$ must totally spread because $H^{(s_0)} $ is not allowed to have zero energy bound states. Since $H^{(s_0)} $ is the Hamiltonian of $N-1$ particles 
by the induction assumption it follows that 
\begin{equation}
 (\phi_k , \chi_L (r_1 -r_2) \phi_k) \to 0  , 
\end{equation}
which contradicts (\ref{mishka}).
\end{proof}

\begin{lemma}\label{lem:7}
Let $f_n (x)\in D(H_0) \subset L^2 (\mathbb{R}^{3N-3})$ and $\| f_n \| + \|H_0
f_n \| \leq 1$. Suppose that $\| \chi_{\{x| |r_i - r_j| < q\}} f_n \| \to 0$ for
some
fixed $i \neq j$ and any $q > 0$.
Then $(f_n, F(r_i - r_j ) f_n ) \to 0$ for any given $F \in L^2 (\mathbb{R}^3) + L^\infty_\infty (\mathbb{R}^3)$.
\end{lemma}
\begin{proof}
Obviously, it suffices to consider $F \in L^2 (\mathbb{R}^3)$. 
\begin{eqnarray}
\fl  \| |F|^{1/2} f_n \| \leq  (\chi_{\{x| |r_i - r_j| < q\}} f_n, |F|
f_n)  + (\chi_{\{x| |r_i - r_j| \geq  q\}} f_n, |F| f_n) \\
\fl = (\chi_{\{x| |r_i - r_j| < q\}} f_n, |F|  f_n)  + (f_n, \chi_{\{x| |r_i
- r_j| \geq  q\}} |F|(H_0 +1)^{-1} (H_0 + 1)f_n) \\
\fl \leq \| \chi_{\{x| |r_i - r_j| < q\}} f_n\| \; \||F|f_n \| + \Bigl\|
\chi_{\{x| |r_i - r_j| \geq  q\}} |F| (H_0 +1)^{-1} \Bigr\| .  \label{l7:1}
\end{eqnarray}
The first term in (\ref{l7:1}) goes to zero
because
$|F(r_i -r_j)|$ is relatively $H_0$ bounded. 
The second term is an operator norm, which can be made as small as pleased by setting $q$ large
enough, see Lemma~5 in \cite{1}. 
\end{proof}

\ack 

The author would like to thank Prof. Walter Greiner for the warm hospitality at FIAS.

\end{document}